%% file: main.tex
\documentclass[11pt]{article} 

\usepackage[margin=1in]{geometry} 
\usepackage{libertine}
\usepackage{float}
\usepackage{todonotes}
\usepackage[english]{babel}
\usepackage{palatino}
\usepackage{csquotes}

\usepackage{comment,xspace}
\usepackage[dvipsnames]{xcolor} 
\definecolor{ptblue}{RGB}{15,76,129} 
\definecolor{ptemerald}{HTML}{009473} 
\definecolor{bluegray}{rgb}{0.4, 0.6, 0.8}
\definecolor{ptilluminating}{HTML}{F5DF4D} 
\definecolor{ptgray}{HTML}{939597} 
\usepackage{cprotect}
\usepackage[normalem]{ulem} 

\usepackage[ruled,linesnumbered,vlined]{algorithm2e}
\usepackage{algorithmic}
\usepackage{comment,xspace,enumitem}
\usepackage{booktabs}
\usepackage{wrapfig}

\usepackage{changepage}
\usepackage{float}

\SetCommentSty{mycommfont}
\usepackage{booktabs,multirow,subcaption}

\usepackage[breakable, theorems, skins]{tcolorbox}

\usepackage{amsmath,amsfonts,amssymb,amsthm,mathtools,thmtools} 
\DeclareMathOperator*{\argmax}{arg\,max}
\DeclareMathOperator*{\argmin}{arg\,min}

\SetCommentSty{mycommentfont}

\usepackage{thm-restate}

\usepackage[backend=biber,style=alphabetic,natbib=true,url=false,doi=true,maxbibnames=99]{biblatex}
\usepackage{hyperref}
\hypersetup{
linktocpage,
colorlinks=true,
citecolor=cobalt, 
urlcolor=ptblue, 
linkcolor=cobalt, 
}
\usepackage{cleveref} 

\usepackage{tikz} 
\theoremstyle{plain}
\newtheorem{theorem}{Theorem}[section]
\newtheorem{corollary}[theorem]{Corollary}
\newtheorem{proposition}[theorem]{Proposition}
\newtheorem{lemma}[theorem]{Lemma}

\theoremstyle{definition}
\newtheorem{definition}[theorem]{Definition}

\newtheorem*{theorem*}{Theorem}

\theoremstyle{remark}

\definecolor{cobalt}{rgb}{0.0, 0.28, 0.67}

\makeatletter 
\newcommand{\EF}[1]{\if\relax\detokenize\expandafter{\@firstofone#1{}}\relax EF\xspace\else EF#1\fi}
\makeatother

\newcommand{\calA}{\mathcal{A}}
\newcommand{\calB}{\mathcal{B}}

\newcommand{\calI}{\mathcal{I}}

\newcommand{\A}{\mathcal{A}}
\newcommand{\alloc}{\mathcal{A}}
\newcommand{\allocn}{\alloc = (A_1,\dots A_n)}
\newcommand{\GenInstance}{\calI=\langle N, M,\{v_i\}_{i \in N}\rangle}
\newcommand{\Partition}{\textup{\textsc{Partition}}}
\newcommand{\RestrictedPartition}{\textup{\textsc{Restricted-Partition}}}

\newcommand{\EQ}[1]{\ifstrempty{#1}{\textrm{\textup{EQ}}}{\textrm{\textup{EQ{$#1$}}}}}

\title{\bfseries
The Landscape of Almost Equitable Allocations}

\author{
\textbf{Hadi Hosseini} \\ 
Pennsylvania State University, USA\\ 
\texttt{hadi@psu.edu}
\and 
\textbf{Vishwa Prakash HV}\\
Chennai Mathematical Institute, Chennai\\
\texttt{vishwa@cmi.ac.in}
\and 
\textbf{Aditi Sethia}\\
Indian Institute of Science, Bangalore\\
\texttt{aditisethia@iisc.ac.in} 
\and
\textbf{Jatin Yadav}\\
Indian Institute of Technology, Delhi\\
\texttt{jatin.yadav@cse.iitd.ac.in}
} 

\date{}

\begin{document}
\maketitle
\input{sections/abstract.tex}

\input{sections/introduction.tex}

\input{sections/preliminaries.tex}

\input{sections/two_agents.tex}

\input{sections/hardness.tex}

\input{sections/submodular_dmonotone.tex}
\input{sections/nonneg.tex}

\input{sections/identical_subadd.tex}
\input{sections/eq_transform.tex}
\input{sections/conclusion.tex}
\input{sections/acknowledgements.tex}
\printbibliography[title={References}]

\appendix
\input{sections/appendix.tex}


\end{document}

%% file: sections/abstract.tex
\begin{abstract}
Equitability is a fundamental notion in fair division which requires that all agents derive equal value from their allocated bundles. We study, for general (possibly non-monotone) valuations, a popular relaxation of equitability known as \textit{equitability up to one item} (\EQ{1}). An \EQ{1} allocation may fail to exist even with additive non-monotone valuations; for instance, when there are two agents, one valuing every item positively and the other negatively. This motivates a mild and natural assumption: all agents agree on the sign of their value for the \emph{grand bundle}. Under this assumption, we prove the existence and provide an efficient algorithm for computing \EQ{1} allocations for two agents with general valuations. When there are more than two agents,  we show the existence and polynomial-time computability of \EQ{1} allocations for valuation classes beyond additivity and monotonicity, in particular for (1) \emph{doubly monotone} valuations and (2) \emph{submodular} (resp. \emph{supermodular}) valuations where the value for the grand bundle is \textit{non-negative} (resp. \textit{nonpositive}) for all agents. Furthermore, we settle an open question of Bil\`o et al. by showing that an \EQ{1} allocation always exists for non-negative (resp. nonpositive) valuations, i.e., when every agent values each subset of items non-negatively (resp. nonpositively).
Finally, we complete the picture by showing that for general valuations with more than two agents, \EQ{1} allocations may not exist even when agents agree on the sign of the grand bundle, and that deciding the existence of an \EQ{1} allocation is computationally intractable.
\end{abstract}

%% file: sections/introduction.tex
\section{Introduction}\label{sec:intro}

A central problem in multiagent systems concerns the fair allocation of resources, tasks, or items among agents with heterogeneous preferences. At its core lies a simple yet fundamental question: \emph{how should such resources be divided fairly?} The field of fair division provides a formal framework for studying this question and has introduced variety of fairness axioms, particularly in settings involving indivisible items.
Among these, \textit{equitability} \cite{Dubins1961} is a compelling and well-studied notion, requiring that all agents derive equal subjective value from their allocated bundles. Motivated by empirical evidence \cite{Herreiner2009,HERREINER2010238}, equitability captures an \textit{interpersonal} notion of fairness by emphasizing equality in experienced happiness among agents. 

The vast majority of existing work assumes that agents' valuations are \textit{monotonic}, i.e., receiving additional items monotonically increases (in case of goods) or decreases (in case of chores) an agent's utility.
However, many realistic settings involve \textit{non-monotonic} or even \textit{non-additive} valuations. Moreover, whether the inclusion of a given item increases or decreases an agent's utility, could depend on the subset of items she already owns. For example, consider a firm looking to hire a new employee. While this individual might be highly valuable in isolation, the marginal utility of bringing them on could be negative if the team already has several people with the same expertise, or if the new hire’s work style is not aligned with the team. The value the firm derives from the employee depends not only on their abilities, but also on the existing team composition and dynamics.
Similarly, in parallel processing, adding more processors does not always improve performance. Initially, tasks complete faster as workload divides, but beyond a point, communication overhead and synchronization costs dominate --- each new processor adds more waiting than work. 


Motivated by these examples, we study the fairness notion of equitability under the most \textit{general valuation} functions, extending beyond standard monotonic or additive assumptions.
Our focus is on its prominent relaxation, \textit{equitability up to one item} (\EQ{1}), which requires that any inequality between two agents can be eliminated by removing at most one item from an agent's bundle.

When valuations are additive and monotone, an \EQ{1} allocation is known to exist when all items are goods for all agents \cite{FreemanGoods19}, and likewise when all items are chores \cite{FreemanChores20}. Even in \emph{mixed-manna} settings, if agents agree on whether each item is a good or a chore, an \EQ{1} allocation is guaranteed to exist \cite{hosseini2025}.\footnote{Such valuations are called \emph{objective}, meaning that no item is valued positively by some agents and negatively by others.} In contrast, if agents disagree on the sign of items' marginal values, an \EQ{1} allocation need not exist. For instance, consider two agents with additive valuations over more than two items: one assigns a value of \(1\) to every item, while the other assigns \(-1\). In this case, no complete allocation can satisfy \EQ{1}.

Motivated by this impossibility, we introduce a mild assumption: all agents agree on the sign of the \textit{grand bundle}, that is, whether the collection of all items is overall desirable or undesirable. Under this assumption, we ask the following theoretical questions:
\textit{For which classes of valuation functions is an \EQ{1} allocation guaranteed to exist, and do the corresponding decision problems admit efficient algorithms?}

\subsection{Our Contributions}
We study the existence and computation of \EQ{1} allocations under general valuations and several subclasses, assuming that all agents agree on the sign of their value for the grand bundle. Figure~\ref{fig:venn_comparison} presents the resulting landscape of existence across major valuation classes. Specifically, we examine  \EQ{1} allocations for (i) general valuations without additional assumptions, (ii) doubly monotone valuations where items can be partitioned into goods and chores, (iii) submodular and supermodular valuations, (iv) subadditive and superadditive valuations, and (v) nonnegative and nonpositive valuations.

For clarity of exposition, we first present all our results under the assumption that all agents value the grand bundle \emph{nonnegatively}. We then establish a technical result (\Cref{thm:nonpos_to_nonneg}) that extends these findings to the complementary case in which the grand bundle is valued nonpositively (\Cref{sec:nonpos}), thereby providing a complete characterization of \EQ{1} allocations.

\paragraph{General valuations.} We show that, for three or more agents with supermodular valuations (and therefore arbitrary general valuations), an \EQ{1} allocation may not exist—even when all agents value the grand bundle nonnegatively. Moreover, the corresponding decision problem is NP-complete (\Cref{thm:hard_supermodular}).  In contrast, for two agents, an \EQ{1} allocation always exists and can be computed in polynomial time whenever both agents agree on the sign of the grand bundle (\Cref{thm:two_agents}, \Cref{thm:two_agents_nonpos}).

\paragraph{Doubly monotone and submodular valuations.} We introduce a new class of valuations---those satisfying the \emph{marginal witness property}.\footnote{Formal definition in \Cref{sec:submodular_dmonotone}.} For this class, when all agents value the grand bundle nonnegatively, an \EQ{1} allocation always exists and can be computed in polynomial time. Since both \emph{doubly monotone} (\Cref{thm:doublymonotone}) and \emph{submodular} (\Cref{thm:submodular}) valuations also satisfy this property and hence are contained within this class, our results immediately imply existence and efficient computation of \EQ{1} allocations for these two important valuation classes.
 
\paragraph{Nonnegative valuations.}A particular valuation class that has recently received attention in the literature \cite{bilò2025,barman2025eq} is that of \emph{nonnegative valuations}, where every subset of items has nonnegative value, implying that agents never view any bundle as a net loss. Within this class, recent works study weaker relaxations of equitability obtained via rounding equitable cake divisions and fixed-point arguments \cite{bilò2025,barman2025eq}. Specifically, \cite{bilò2025} achieve equitability after removing at most two items (at most one from each bundle), while \cite{barman2025eq} guarantee equitability after the removal or addition of at most three items. Both works leave the existence of \EQ{1} allocations as an intriguing open question. In this work, we resolve this question positively by proving, through a combinatorial approach, that \EQ{1} allocations always exist (\Cref{thm:nonneg}).

\paragraph{Identical subadditive valuations.} 
We show the existence of \EQ{1} allocations for identical subadditive valuations when agents value the grand bundle nonnegatively (\Cref{thm:identical_subadd}). 
This finding strengthens the result of~\cite{barman2025eq}. 

Since \EQ{1} and \EF{1}\footnote{An allocation is envy-free (EF) if no agent values another agent's bundle more than its own. Analogous to \EQ{1}, \EF{1} is a relaxation where envy, if any, can be eliminated by removing at most one item.} coincide under identical valuations, this result immediately implies the existence of \EF{1} under identical subadditive valuations (\Cref{cor:identical}).


\begin{figure*}[h]
    \centering
    \begin{subfigure}[b]{0.49\textwidth}
        \centering
        \scalebox{0.65}{\input{figs_tables_algo/venn_pos}}
        \caption{Agents value the grand bundle nonnegatively.
        }
        \label{fig:venn_pos}
    \end{subfigure}
    \hspace{-0.2cm}
    \begin{subfigure}[b]{0.49\textwidth}
        \centering
        \scalebox{0.65}{
        \input{figs_tables_algo/ven_neg}}
        \caption{Agents value the grand bundle nonpositively.}
        \label{fig:venn_neg}
    \end{subfigure}
    \caption{A pictorial representation of various valuation classes, their intersections, and an overview of our results. $\checkmark$ implies existence of \EQ{1} allocations, $\times$ implies non-existence, $?$ implies that the existence is an open question. For two agents, we show the existence and efficient computation of \EQ{1} allocations under \textit{general valuations} when the grand bundle is valued nonnegatively (resp. nonpositively) (\Cref{thm:two_agents}).
    }
    \label{fig:venn_comparison}
\end{figure*}
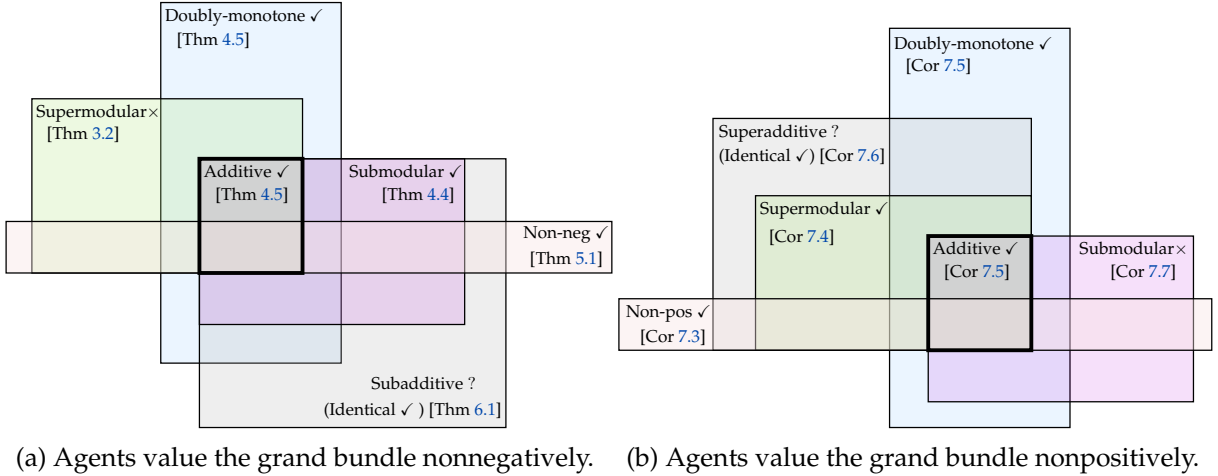

\subsection{Related Work}

\paragraph{Divisible Items.} When the items are divisible, \citeauthor{Dubins1961}~\cite{Dubins1961} showed that an equitable division always exists for additive valuations. \citeauthor{CECHLAROVA2013}~\cite{CECHLAROVA2013}  showed that connected EQ allocations exist for valuation functions that are nonnegative, non-decreasing, and continuous, and in such settings, nearly equitable allocations can be computed efficiently \cite{CECHLAROVA2012}. 
Ch\`eze \cite{CHEZE2017} gave a simpler and shorter existence proof based on a fixed-point result (Borsuk-Ulam Theorem). \citeauthor{SK23}~\cite{SK23} showed the existence for general identical valuations.
\citeauthor{AumannDombb2015}~\cite{AumannDombb2015} showed that there is a connected equitable division that also maximizes the egalitarian welfare.
Recently, \citeauthor{BSV2025}~\cite{BSV2025} showed the existence of \EQ{} allocations for nonnegative valuations using Sperner's Lemma. Their proof also generalizes to subclasses of possibly negative valuations like identical valuations and single-peaked valuations. The above results are only existential, \citeauthor{PW17}~\cite{PW17} showed that there is no bounded runtime algorithm for finding an equitable division, even without the connectedness constraint.

\paragraph{Indivisible Items.} When the items are indivisible, an exact EQ allocation may not exist, but \EQ{1} allocations are known to exist and are efficiently computable for monotone additive valuations \cite{FreemanGoods19,FreemanChores20}. Under additive valuations, several works have also studied equitability in conjunction with efficiency guarantees \cite{FreemanGoods19,GargM24}, and welfare trade-offs \cite{CKKK2012,SCDV23,BMSV23}.

\paragraph{Equitability Beyond Additivity and/or Monotonicity.} Recently, approximate equitability has been considered in the non-additive and/or non-monotone settings. \citeauthor{barman2024nearly}~\cite{barman2024nearly} considered a stronger relaxation, EQX, where the removal of \emph{any} good from the rich agent's bundle or the removal of \emph{any} chore from the poor agent's bundle gets rid of the inequity. They showed that EQX exists for monotone valuations and can be efficiently computed for \emph{weakly well-layered valuations.} For non-monotone valuations, they showed the existence of EQX allocations for the case of two agents with additive valuations where each
item is either a good for both agents or a chore for both agents. \citeauthor{hosseini2025}~\cite{hosseini2025} showed the existence and efficient computation of \EQ{1} allocations when every item is valued at $\{-\alpha, 0, \alpha\}$ and the valuations are additive. In \cite{barman2025eq}, \citeauthor{barman2025eq} study a weaker notion of equitability under nonnegative valuations. They show the existence of allocations where the difference between the utilities of any two agents can be eliminated by the removal or addition of at most three items. \citeauthor{bilò2025}~\cite{bilò2025} prove a stronger result by showing the existence of $\EQ{1}_g^c$ allocations\footnote{An allocation is $\EQ{1}_g^c$ if for every pair of agents equitability can be guaranteed by the removal of up to removal of one item \emph{each} from their bundles \cite{bilò2025}. } for nonnegative valuations.

\paragraph{Envy-Freeness Beyond Additivity and/or Monotonicity.}
\EF{1} allocations have been known to exist for general monotone valuations and can be computed efficiently with envy-cycle eliminations \cite{LMM04}. For non-monotone but additive valuations, \citeauthor{Aziz2021}~\cite{Aziz2021} showed the existence and efficient computation of \EF{1} by double round-robin algorithm. \citeauthor{Bhaskar2020}~\cite{Bhaskar2020} extended this to doubly monotone valuations. Recently, \citeauthor{BKPR25}~\cite{BKPR25} showed that \EF{1} always exists for (i) identical trilean valuations and (ii) a newly introduced class of valuations--\emph{separable single-peaked (SSP)
valuations.} \citet{hosseini25} considered \EF{1} allocation of the vertices of a graph among the agents, where the value of a bundle is determined by its cut value, capturing settings where valuations are inherently non-monotonic.

%% file: figs_tables_algo/venn_pos.tex
\tikzset{every picture/.style={line width=0.75pt}} 

\begin{tikzpicture}[x=0.75pt,y=0.75pt,yscale=-1,xscale=1]

\draw  [fill={rgb, 255:red, 214; green, 235; blue, 255 }  ,fill opacity=0.47 ] (260,10) -- (400,10) -- (400,290) -- (260,290) -- cycle ;
\draw  [fill={rgb, 255:red, 184; green, 184; blue, 184 }  ,fill opacity=0.23 ] (290,131.32) -- (527.8,131.32) -- (527.8,340) -- (290,340) -- cycle ;
\draw  [fill={rgb, 255:red, 189; green, 16; blue, 224 }  ,fill opacity=0.13 ] (290,131.32) -- (495.8,131.32) -- (495.8,260) -- (290,260) -- cycle ;
\draw  [fill={rgb, 255:red, 184; green, 233; blue, 134 }  ,fill opacity=0.25 ] (160,84.72) -- (370,84.72) -- (370,220) -- (160,220) -- cycle ;
\draw  [fill={rgb, 255:red, 249; green, 229; blue, 229 }  ,fill opacity=0.39 ] (140,180) -- (610,180) -- (610,220) -- (140,220) -- cycle ;
\draw  [line width=2.25]  (290,131.32) -- (370,131.32) -- (370,220) -- (290,220) -- cycle ;

\draw (292,134.32) node [anchor=north west][inner sep=0.75pt]  [font=\small] [align=left] {Additive $\checkmark$ };
\draw (405,134) node [anchor=north west][inner sep=0.75pt]  [font=\small] [align=left] {Submodular $\checkmark$};
\draw (422,299) node [anchor=north west][inner sep=0.75pt]  [font=\small] [align=left] {Subadditive $\displaystyle ?$};
\draw (540,182) node [anchor=north west][inner sep=0.75pt]  [font=\small] [align=left] {Non-neg $\checkmark$};
\draw (262,13) node [anchor=north west][inner sep=0.75pt]  [font=\small] [align=left] {Doubly-monotone $\checkmark$ };
\draw (545,202) node [anchor=north west][inner sep=0.75pt]  [font=\small] [align=left] {[Thm~\ref{thm:nonneg}]};
\draw (430,152) node [anchor=north west][inner sep=0.75pt]  [font=\small] [align=left] {[Thm~\ref{thm:submodular}]};
\draw (271,32) node [anchor=north west][inner sep=0.75pt]  [font=\small] [align=left] {[Thm~\ref{thm:doublymonotone}]};
\draw (301,152) node [anchor=north west][inner sep=0.75pt]  [font=\small] [align=left] {[Thm~\ref{thm:doublymonotone}]};
\draw (385,319) node [anchor=north west][inner sep=0.75pt]  [font=\small] [align=left] {(Identical $\checkmark$ ) [Thm~\ref{thm:identical_subadd}]};
\draw (162,87.72) node [anchor=north west][inner sep=0.75pt]  [font=\small] [align=left] {Supermodular$\displaystyle \times $};
\draw (170,105) node [anchor=north west][inner sep=0.75pt]  [font=\small] [align=left] {[Thm~\ref{thm:hard_supermodular}]};

\end{tikzpicture}

%% file: figs_tables_algo/ven_neg.tex
\tikzset{every picture/.style={line width=0.75pt}} 

\begin{tikzpicture}[x=0.75pt,y=0.75pt,yscale=-1,xscale=1]

\draw  [fill={rgb, 255:red, 214; green, 235; blue, 255 }  ,fill opacity=0.47 ] (270,00) -- (410,00) -- (410,310) -- (270,310) -- cycle ;
\draw  [fill={rgb, 255:red, 184; green, 184; blue, 184 }  ,fill opacity=0.23 ] (133,70) -- (380,70) -- (380,250) -- (133,250) -- cycle ;
\draw  [fill={rgb, 255:red, 189; green, 16; blue, 224 }  ,fill opacity=0.13 ] (300,161.32) -- (505.8,161.32) -- (505.8,290) -- (300,290) -- cycle ;
\draw  [fill={rgb, 255:red, 184; green, 233; blue, 134 }  ,fill opacity=0.25 ] (166,130) -- (380,130) -- (380,250) -- (166,250) -- cycle ;
\draw  [fill={rgb, 255:red, 249; green, 229; blue, 229 }  ,fill opacity=0.39 ] (60,210) -- (520,210) -- (520,250) -- (60,250) -- cycle ;
\draw  [line width=2.25]  (300,161.32) -- (380,161.32) -- (380,250) -- (300,250) -- cycle ;

\draw (302,164.32) node [anchor=north west][inner sep=0.75pt]  [font=\small] [align=left] {Additive $\checkmark$ };
\draw (416,164) node [anchor=north west][inner sep=0.75pt]  [font=\small] [align=left] {Submodular$\displaystyle \times $};
\draw (136.2,73) node [anchor=north west][inner sep=0.75pt]  [font=\small] [align=left] {Superadditive $\displaystyle ?$};
\draw (64,213) node [anchor=north west][inner sep=0.75pt]  [font=\small] [align=left] {Non-pos $\checkmark$};
\draw (272,05) node [anchor=north west][inner sep=0.75pt]  [font=\small] [align=left] {Doubly-monotone $\checkmark$};
\draw (75,232) node [anchor=north west][inner sep=0.75pt]  [font=\small] [align=left] {[Cor~\ref{thm:nonpos}]};
\draw (440,182) node [anchor=north west][inner sep=0.75pt]  [font=\small] [align=left] {[Cor~\ref{thm:hard_submodular}]};
\draw (281,22) node [anchor=north west][inner sep=0.75pt]  [font=\small] [align=left] {[Cor~\ref{thm:doublymonotone_nonpos}]};
\draw (311,182) node [anchor=north west][inner sep=0.75pt]  [font=\small] [align=left] {[Cor~\ref{thm:doublymonotone_nonpos}]};
\draw (136,92) node [anchor=north west][inner sep=0.75pt]  [font=\small] [align=left] {(Identical $\checkmark$) [Cor~\ref{thm:superadditive_nonpos}]};
\draw (168,133) node [anchor=north west][inner sep=0.75pt]  [font=\small] [align=left] {Supermodular $\checkmark$};
\draw (175,155) node [anchor=north west][inner sep=0.75pt]  [font=\small] [align=left] {[Cor~\ref{thm:supermodular_nonpos}]};

\end{tikzpicture}

%% file: sections/preliminaries.tex
\section{Preliminaries}\label{sec:prelim}

\paragraph{The Setting.} An instance of a fair division problem is given by a tuple $\GenInstance$, where $N$ is a set of $n \in \mathbb{N}$ agents and $M$ is a set of $m \in \mathbb{N}$ indivisible items. For the sake of brevity, we will often assume $N = \{1, 2, \ldots n\}$.
Every agent $i$ has a valuation function $v_i : 2^M \to \mathbb{R}$, associating a real value to every set $S \subseteq M$ of items, denoted by $v_i(S)$. Throughout the paper, we assume a polynomial time oracle access to the agent valuation functions.
An allocation $\allocn$ is a partition of $M$ into $n$ bundles, where $A_i$ denotes the bundle allocated to agent $i$.

\paragraph{Valuation Classes.} The marginal value of an item $o$, for an agent \(i\), with respect to a subset $S \subseteq M$ is defined as $v_i(S|e) = v_i(S \cup \{e\}) - v_i(S)$. We consider non-monotone valuations, meaning that, for an agent, an item may have either positive, negative, or zero marginal value with respect to a given subset of items. We additionally assume that the sign of the valuation of the \textit{grand bundle} (the entire set $M$ of items) is the same for all agents. Formally, either $v_i (M) \ge 0$ for all agents $i \in N$, or $v_i (M) \le 0$ for all agents $i \in N$. This assumption is quite natural and reasonable; without it,\EQ{1} allocations may fail to exist even in simple cases with just two items and two agents with additive valuations, one valuing each item at 1, another at -1. Here, no (complete) allocation can satisfy\EQ{1}.

We next define the classes of valuations that we consider in this work. 

\begin{enumerate}
    \item \emph{Additive}: A valuation function $f$ is said to be \emph{additive} if the value of a bundle $S$ of items is equal to the sum of the values of the items in $S$. Formally, $f(S) = \sum_{e \in S} f(e)$ for all subsets $S \subseteq M$.

    \item \emph{Submodular / Supermodular}: A valuation function $f$ is said to be \emph{submodular} if it satisfies the property of diminishing marginal value. Formally, for subsets $S, T \subseteq M$ such that $S \subseteq T$, and any item $e \notin T$, we have $f(S \cup e) - f(S) \ge f(T \cup e) - f(T)$. Likewise, a valuation function $f$ is said to be \emph{supermodular} if it satisfies the property of increasing marginal value. Formally, for subsets $S, T \subseteq M$ such that $S \subseteq T$, and any item $e \notin T$, we have $f(S \cup e) - f(S) \le f(T \cup e) - f(T)$.

    \item \emph{Subadditive}: We say that a valuation function $f$ is said to be \emph{subadditive} if for any two \textbf{disjoint} subsets $S, T \subseteq M$, we have $f(S \cup T) \le f(S) + f(T)$. Note that, this is a weaker condition than the standard definition of subadditivity, which requires the above condition to hold for any two (not necessarily disjoint) subsets $S, T \subseteq M$.\footnote{Putting $S = T$ in the standard definition, we get $f(S) \le 2 f(S)$, or equivalently, $f(S) \ge 0$ for all subsets $S \subseteq M$. Thus, the standard definition of subadditivity implies that the valuation function is nonnegative. However, in this work, we consider more general subadditive valuation functions that may take negative values as well.}

    \item \emph{Superadditive}: A valuation function $f$ is said to be \emph{superadditive} if for any two \textbf{disjoint} subsets $S, T \subseteq M$, we have $f(S \cup T) \ge f(S) + f(T)$.

    \item \emph{Doubly Monotone}: A valuation function $f$ is said to be \emph{doubly monotone}, if the set of items $M$ can be partitioned into \emph{goods} $G$ and \emph{chores} $C$, such that $M = G \sqcup C$ and for all $S \subseteq M$, we have $f(S \cup g) \geq f(S)$ for each good $g \in G$ and $f(S \cup c) \leq f(S)$ for each chore $c \in C$. It is important to note that in an instance of fair division with doubly monotone valuations, different agents may have different sets of goods and chores.
    
    \item \emph{nonnegative / nonpositive}: A valuation function $f$ is said to be \emph{nonnegative} if the value of any subset of items is a nonnegative real number. Formally, if $f(S) \ge 0$ for all subsets $S \subseteq M$. Likewise, a valuation function $f$ is said to be \emph{nonpositive} if the value of any subset of items is a nonpositive real number. Formally, if $f(S) \le 0$ for all subsets $S \subseteq M$. 

\end{enumerate}

The general containment of valuation classes is as follows.
$$\text{Additive} \subset \text{Submodular} \subset \text{Subadditive}$$ 
$$ \text{Additive} \subset \text{Doubly Monotone}$$ 
$$\text{Additive} \subset \text{Supermodular} \subset \text{Superadditive}$$

\paragraph{Fairness Notion.}
An allocation $\allocn$ is said to be \textit{equitable} (EQ) if all the agents derive equal value from their respective bundles, that is, for every pair of agents $i$ and $j$, we have $v_i(A_i) = v_j(A_j)$. An EQ allocation need not always exist. Thus, we consider a popular relaxation of EQ, called as\EQ{1}. An allocation $\mathcal{A}$ is said to be \textit{equitable up to one item} (EQ1) if for every pair of agents $i$ and $j$ such that $v_i(A_i) < v_j(A_j)$, either there exists some item $g$ in $A_j$ such that $v_i(A_i) \geq v_j(A_j \setminus \{g\})$ or there exists some item $c$ in $A_i$ that satisfies $v_i(A_i \setminus \{c\}) \geq v_j(A_j)$.

We now introduce a stronger notion of \emph{lower\EQ{1} witness} that implies\EQ{1}. This captures\EQ{1} through the existence of a single \emph{witness value} that represents an equitable value level. Intuitively, an allocation admits a \emph{lower\EQ{1} witness} if all agents value their bundles at least as much as some common threshold $\theta$, and the valuation of each agent could fall below that threshold by the agent losing at most one item.

\begin{definition}[Lower\EQ{1} Witness]\label{def:lower-eq1-witness}
An allocation $\mathcal{A}$ is said to \emph{admit a lower\EQ{1} witness} $\theta \in \mathbb{R}$ if the following conditions hold:
\begin{enumerate}
    \item For every agent $i$, $v_i(A_i) \ge \theta$, and
    \item For every agent $i$, either $v_i(A_i) = \theta$, or there exists an item $g \in A_i$ such that $v_i(A_i \setminus \{g\}) \le \theta$.
\end{enumerate}
\end{definition}

The existence of a lower witness guarantees that all agents are “almost” at a common value level—within the change caused by one item. Hence, it immediately implies \EQ{1}.

\begin{proposition}\label{prop:witness-implies-eq1}
If an allocation $\mathcal{A}$ admits a lower\EQ{1} witness, then $\mathcal{A}$ is\EQ{1}.
\end{proposition}

\begin{definition}[Rich and Poor Agents]
    Given an allocation  $\allocn$, an agent $i$ is said to be a \emph{rich} agent if $v_i(A_i) \geq v_j(A_j)$ for all $j \in N$. Similarly, an agent $i$ is said to be a \emph{poor} agent if $v_i(A_i) \leq v_j(A_j)$ for all $j \in N$.
\end{definition}

\subsection{Organization of the Paper.}
We first consider the case of nonnegatively valued grand bundle in Sections \ref{sec:general_agents}-\ref{sec:identical_subadditive}.
We begin with the case of general valuations in \Cref{sec:general_agents}, followed by doubly monotone valuations and submodular valuations in \Cref{sec:submodular_dmonotone}, nonnegative valuations in \Cref{sec:nonneg}, and identical subadditive valuations in \Cref{sec:identical_subadditive}. 
The analogous case where the grand bundle is valued nonpositively appears in~\Cref{sec:nonpos}.

%% file: sections/two_agents.tex
\section{General Valuations}
\label{sec:general_agents}
In this section, we consider general valuations with nonnegative grand bundles ($v_i(M) \geq 0~\forall~i \in N$). We show that for two agents, an \EQ{1} allocation always exists and can be efficiently computed. 
On the other hand, for more than two agents, we exhibit non-existence and prove that its corresponding decision problem is NP-complete.

~\subsection{Two Agents: An Existence Result}

We first present an efficient algorithm for the case of two agents.

\begin{theorem}
\label{thm:two_agents}
Given a fair division instance $\GenInstance$ with two agents such that the grand bundle is valued nonnegatively by both the agents, an \EQ{1} allocation always exists and can be computed efficiently.
 \end{theorem}

 \paragraph{Algorithm Overview.} We begin by checking a trivial case: if agent 1 assigns zero total value to all items, give everything to agent 1 and we are done. Otherwise, imagine the items laid out in some fixed order on a table. A ``knife'' sweeps from left to right, moving items from the right pile to the left pile one at a time. After each move, the left pile is what agent 1 would receive and the right pile is what agent 2 would receive. We keep moving the knife one item at a time until, for the first time, agent 1’s value for the left pile becomes strictly larger than agent 2’s value for the right pile. Call that first stopping point position $i$. At that moment we have two neighboring cut points to consider: just after item $i - 1$ or just after item $i$. It turns out that at least one of these two allocations is \EQ{1}, and we can check which one in polynomial time. The full details are in Algorithm~\ref{alg:two_agents}.

\input{figs_tables_algo/two_agents_algorithm.tex}

    \begin{proof}[Proof of \Cref{thm:two_agents}]
    
        If $v_1(M)=0$, the algorithm returns the allocation $\calA$ such that $A_1 = M$ and $A_2 = \emptyset$ which is clearly \EQ{1}. So assume $v_1(M)>0$. Let $f(t)=v_1(S_t)-v_2(T_t)$ for all $t \in \{0, 1, \ldots m\}$. Then, we have $f(0)=-v_2(M)\le 0$ and $f(m)=v_1(M)>0$. Our algorithm finds the smallest $i \in [m]$ such that $f(i) > 0$. By the minimality of $i$, and the fact that $f(0) \le 0$, we have $f(i - 1) \le 0$. Thus, we have the following two inequalities:
    \begin{enumerate}
        \item $v_1(S_i) > v_2(T_i)$
        \item $v_1(S_{i-1}) \le v_2(T_{i-1})$
    \end{enumerate}

    The algorithm sets $\calA=(S_i,T_i)$ and $\calB=(S_{i-1},T_{i-1})$, returning $\calA$ if it is \EQ{1}, otherwise $\calB$.
    
    If $\calA$ is \EQ{1} we are done. Otherwise removing $e_i$ from $S_i$ does not eliminate the inequality. Thus
    $v_1(S_{i-1}) = v_1(S_i\setminus\{e_i\}) > v_2(T_i)$. Now, in the allocation $\calB$, $v_1(S_{i-1}) \le v_2(T_{i-1})$ by (2) and $v_2(T_{i-1} \setminus e_i) = v_2(T_i) < v_1(S_{i-1})$. Hence $\calB$ is \EQ{1}.
    
    The runtime is clearly polynomial since checking the inequalities can be done in constant time.
    
    \end{proof}

%% file: figs_tables_algo/two_agents_algorithm.tex
\begin{algorithm}[ht]
\SetAlgoLined
\DontPrintSemicolon
\KwIn{An instance $\GenInstance$ with two agents such that $v_i(M) \geq 0$ for $i \in \{1, 2\}$.}
\KwOut{An \EQ{1} allocation.}
\If{$v_1(M)=0$}{
    \Return $(M,\emptyset)$
}

Order the items arbitrarily as $e_1,e_2,\dots,e_m$.\;
\For{$t \in \{0, 1, \ldots m\}$}{
    $S_t \gets \{e_1,\dots,e_t\}$\;
    $T_t \gets M \setminus S_t$\;
}
$i \gets \min\{t \in \{1,\dots,m\} \mid v_1(S_t) > v_2(T_t)\}$\;

$\calA \gets (S_i, T_i)$\;
$\calB \gets (S_{i-1}, T_{i-1})$\;

\If{$\calA$ is \EQ{1}}{
    \Return $\calA$
}
\Return $\calB$
\caption{Two Agents.}
\label{alg:two_agents}
\end{algorithm}

%% file: sections/hardness.tex
\subsection{Beyond Two Agents: Non-existence and Hardness Result}\label{sec:hardness}

For more than two agents, we show that even with the assumption of the grand bundle being nonnegatively valued by all the agents, an \EQ{1} allocation may not exist, and the corresponding decision problem is NP-complete. In fact, our result holds even for the restricted class of supermodular valuations.

\begin{restatable}{theorem}{hardness}
\label{thm:hard_supermodular}
For valuations where each agent values the grand bundle nonnegatively, \EQ{1} allocations may not exist. Furthermore, it is NP-Complete to decide whether an \EQ{1} allocation exists, even when the valuations are supermodular.
\end{restatable}

Our proof relies on a reduction from the \Partition{} problem. We defer the proof details to \Cref{sec:hard_appendix}.

%% file: sections/submodular_dmonotone.tex
\section{Marginal Witness Valuations}\label{sec:submodular_dmonotone}

This section establishes the existence and polynomial-time computability of EQ1 allocations for doubly monotone and submodular valuations, under the assumption that every agent values the grand bundle nonnegatively. While our results are stated for these familiar valuation classes, they in fact hold more generally for any valuation class that satisfies the \emph{marginal-witness} property. Informally, a valuation function is said to satisfy the marginal-witness property, if, any non-empty subset having a nonnegative marginal value with respect to a given bundle contains a singleton ``witness'' that also has a nonnegative marginal.

\begin{definition}[Marginal Witness Property] 
A valuation function $v_i: 2^M \rightarrow \mathbb{R}$ is said to satisfy the \emph{marginal witness property} if, for any two disjoint bundles $A, B \subseteq M$ with $B \neq \emptyset$ and $v_i(A \cup B) \geq v_i(A)$, there exists an item $b \in B$ such that $v_i(A \cup \{b\}) \geq v_i(A)$.  
Any valuation function that satisfies the marginal witness property is referred to as a \emph{marginal witness valuation}.
\end{definition}

The main result of this section is stated as the following theorem:

\begin{theorem}\label{thm:marginal_witness}
   Given a fair division instance $\GenInstance$  with marginal-witness valuations such that every agent values the grand bundle nonnegatively, an EQ1 allocation always exists and can be computed in polynomial time.

\end{theorem}

We first present an overview of our algorithm.
\paragraph{Algorithm Overview.} Our algorithm is presented in pseudo-code form as Algorithm~\ref{alg:marginal_witness}. It maintains a partial \EQ{1} allocation $A$ and attempts to allocate either all the remaining items or a single item to some agent. If the set of remaining items $R$ can be allocated to some agent while preserving the \EQ{1} property, the algorithm performs the allocation and terminates. For this purpose, the algorithm performs a couple of early termination checks (lines~\ref{l:mid1} and~\ref{l:mid3}) respectively. If both the checks fail, leveraging the marginal witness property, it identifies an item with a nonnegative marginal value for the poorest agent and assigns it to that agent, ensuring that the \EQ{1} property is maintained. This process continues until all items are allocated. 

\input{figs_tables_algo/marginal_witness}

Before we begin the analysis, let us first argue that the algorithm is well defined, that is, when the termination tests fail, there indeed exists an item $g \in R$, such that $v_p(A_p \cup \{g\}) \geq \mu$. This is true because, the first early termination test (line ~\ref{l:mid1}) fails, and hence $v_p(A_p \cup R) > \mu = v_p(A_p)$. Therefore, the existence of an item with a nonnegative marginal follows from the marginal witness property.

\subsection{Analysis of Algorithm~\ref{alg:marginal_witness}}

Let $A^t$ and $R^t$ denote the partial allocation and the set of remaining items, respectively, at the beginning of the $t^\text{th}$ iteration of the while loop (line~\ref{l:while_loop_begin}). Also, let $p^t, \mu^t$ and $g^t$ denote the variables $p$ (poor agent), $\mu$ (value derived by the poor agent), and $g$ (item corresponding to the marginal witness property) respectively in the $t^\text{th}$ iteration of the while loop. For the sake of completeness, we define $A^0_j = \phi$ for all $j \in N$, $R^0 = M$, $\mu^0 = 0$. The correctness of our algorithm hinges on the following two invariants:

\begin{enumerate}
    \item {\bf Invariant 1: } For every iteration $t$ of the while loop, $\mu^{t-1}$ is a lower \EQ{1} witness for the partial allocation $A^t$. That is, for each $j \in N$:
    \begin{enumerate}
        \item $v_j(A^t_j) \geq \mu^{t-1}$, and
        \item $\displaystyle \min_{S \subseteq A^t_j: |S| \leq 1} v_j(A^t_j \setminus S) \leq \mu^{t-1}$.
    \end{enumerate}
    \item {\bf Invariant 2: } For every iteration $t$, we have $v_j(A_j^t \cup R^t) \geq \mu^{t-1}$ for all $j \in N$.
\end{enumerate}

\begin{lemma}\label{lemma:inv_marg}
    Algorithm~\ref{alg:marginal_witness} maintains {\bf Invariant 1} and {\bf Invariant 2}.
    
\end{lemma}

\begin{proof}
    We proceed via induction on $t$.

    \textbf{Base case.} ($t = 1$): Invariant $1$ is satisfied because $\mu^0 = 0$ and $A^1_j = \emptyset$ for all agents $j \in N$. Invariant $2$ is satisfied because the grand bundle is valued nonnegatively, that is, $v_j(M) \geq 0$ for all $j \in N$. Hence, $v_i(A_i^0 \cup R) = v_i(M) \geq 0 = \mu^0$.

    \textbf{Inductive step.} Consider the $(t+1)^\text{th}$ iteration of the while loop. Since this iteration is reached, the two early termination checks (line~\ref{l:mid1} and line~\ref{l:mid3}) must have failed in the previous ($t^{\text{th}}$) iteration. Note that, $A^{t+1}_j = A^t_j$ for all $j \neq p^t$, and $A^{t+1}_{p^t} = A^t_{p^t} \cup \{g^t\}$. Also, $v_{p^t}(A^{t+1}_p) = v_{p^t}(A^t_{p^t} \cup \{g^t\}) \geq \mu^t = v_{p^t}(A^t_{p^t})$. Hence, agent utilities are non-decreasing, and so is the sequence $(\mu^0, \mu^1, \mu^2, \ldots )$.

    \paragraph{Invariant 1.} Clearly, $v_j(A^{t+1}_j) \geq \mu^{t+1} \geq \mu^t$ for all $j \in N$. Hence, the first requirement of invariant $1$ holds. For each agent $j \neq p^t$, the second requirement of invariant $1$ holds because $A^{t+1}_j = A^t_j$ and $\mu^t \geq \mu^{t-1}$. For agent $p^t$, clearly, $v_{p^t}(A^{t+1}_{p^t} \setminus \{g^t\}) = \mu^t$. Hence the invariant is satisfied in the $(t+1)^\text{th}$ iteration. 

    \paragraph{Invariant 2.} Since the second early termination check (line~\ref{l:mid3}) failed in the $t^{\text{th}}$ iteration, we have $v_j(A^t_j \cup R^t \setminus \{g^t\}) > \mu^t$ for all $j \in N$. Since $R^{t+1} = R^t \setminus \{g^t\}$, this means that $v_j(A^{t+1}_j \cup R^{t+1}) = v_j(A^{t}_j \cup R^{t} \setminus \{g^t\}) \geq \mu^t$ for all $j \neq p^t$. Also, $v_{p^t}(A^{t+1}_{p^t} \cup R^{t+1}) = v_{p^t}((A^t_{p^t} \cup \{g^t\}) \cup (R^t \setminus \{g^t\})) = v_{p^t}(A^t_{p^t} \cup R^t) > \mu^t$, where the last inequality due to the failure of the first early termination check (line~\ref{l:mid1}). Hence, the invariant is satisfied in the $(t+1)^\text{th}$ iteration.
\end{proof}

We will now prove \Cref{thm:marginal_witness}.

\begin{proof}[Proof of \Cref{thm:marginal_witness}]
    Clearly, Algorithm~\ref{alg:marginal_witness} terminates in polynomial time, since each iteration takes polynomial time, and the number of iterations is bounded by  $|M|$. There are three possibilities for the termination state:

    \begin{enumerate}
        \item The first early termination test (line~\ref{l:mid1}) passes in some iteration (say the $t^{\text{th}}$ iteration), that is $\exists i \in N$, such that $v_i(A^t_i \cup R^t) \leq \mu^t$. In this case, we allocate $R^t$ to $i$ and return an allocation $\mathcal{B} = (B_1, \ldots B_n)$ where $B_i = A_i^t \cup R^t$ and $B_j = A_j^t$ for all $j \neq i$. We claim that $\theta = v_i(A^t_i \cup R^t)$ is a lower \EQ{1} witness for the allocation $\mathcal{B}$. Indeed, note that every agent's value under $\mathcal{B}$ is atleast $\theta$ (since $v_i(B_i) = v_i(A_i^t \cup R^t) = \theta$ and for all $j \neq i$, we have $v_j(A^t_j) \geq \mu^t \geq \theta$.) Also, from invariant $2$, we have $\mu^{t-1} \leq \theta$. Now, from invariant $1$, it follows that if $v_j(A_j^t) > \theta \geq \mu^{t-1}$, then there exists $h \in A_j^t$, such that $v_j(A_j^t \setminus \{h\}) \leq \mu^{t-1} \leq \theta$. Therefore, $\mathcal{B}$ is an \EQ{1} allocation.
        
        \item The first early termination test fails but the second early termination test (line~\ref{l:mid3}) passes in some iteration (say the $t^{\text{th}}$ iteration). Hence, there exists $i \in N$, and an item $h \in R^t$, such that $v_i(A^t_i \cup R^t \setminus \{h\}) \leq \mu^t$. In this case, we allocate $R^t$ to agent $i$ and return an allocation $\mathcal{B} = (B_1, \ldots B_n)$ where $B_i = A_i^t \cup R^t$ and $B_j = A_j^t$ for all $j \neq i$. Note that $v_i(B_i) = v_i(A_i^t \cup R^t) \geq \mu^t$, since the first early termination test (line~\ref{l:mid1}) fails. Also, for all agents $j \neq i$, $v_j(B_j) = v_j(A_j^t) \geq \mu^t$ by the definition of $\mu^t$. Hence, $\mu^t$ is a lower bound on the agent utilities. Additionally, from invariant $1$, for each $j \neq i$, $\displaystyle \min_{S \subseteq B_j: |S| \leq 1} v_j(B_j \setminus S) = \min_{S \subseteq A^t_j: |S| \leq 1} v_j(A^t_j \setminus S) \leq \mu^{t-1} \leq \mu^t$. For the agent $i$, clearly $v_i((A^t_i \cup R^t) \setminus \{h\}) \leq \mu^t$. Thus, $\mu^t$ is a lower \EQ{1} witness for $\mathcal{B}$.
        \item The algorithm terminates in line~\ref{l:finish}. Here, the algorithm runs for exactly $|M|$ iterations, allocating a single item in each iteration. It follows from invariant $1$ that $\mu^{|M|}$ is a lower \EQ{1} witness for the final allocation.
    \end{enumerate}
    This completes the proof.
    \end{proof}

As an immediate corollary, Theorems~\ref{thm:submodular} and~\ref{thm:doublymonotone} follow, since submodular and doubly monotone valuations satisfy the marginal witness property; proofs are deferred to Appendix~\ref{apndx:submod_dmon}.

\begin{restatable}{theorem}{submodular}
\label{thm:submodular}
   Given a fair division instance $\GenInstance$ where each agent has a submodular valuation and values the grand bundle nonnegatively, an \EQ{1} allocation always exists and can be computed in polynomial time.
\end{restatable}

\begin{restatable}{theorem}{doublymonotone}
\label{thm:doublymonotone}
   Given a fair division instance $\GenInstance$  where each agent has a doubly monotone valuation and values the grand bundle nonnegatively, an \EQ{1} allocation always exists and can be computed in polynomial time.
\end{restatable}

%% file: figs_tables_algo/marginal_witness.tex
\begin{algorithm}[ht]
\SetAlgoLined
\DontPrintSemicolon
\KwIn{An instance $\GenInstance$ with marginal witness valuations, with $v_i(M) \geq 0~\forall~i \in N$.

} 
\KwOut{An \EQ{1} allocation.}
$A_i \gets \phi$ for all $i \in N$.\\
$R \gets M$\tcp*{Set of remaining items}
\While{$R \neq \phi$}{ \label{l:while_loop_begin}
    Let $p \in \arg \min_{i \in N} v_i(A_i)$ \label{l:begin}\tcp*{A poor agent}
    $\mu \gets v_p(A_p)$\;
    \If{$\exists i \in N$, such that $v_i(A_i \cup R) \leq \mu$}{\label{l:mid1}
        $A_i \gets A_i \cup R$ \label{l:eq1_found_1}\;
        \Return{A}
    } \label{l:mid2}
    \If{$\exists i \in N, h \in R$ such that $v_i(A_i \cup R \setminus \{h\}) \leq \mu$}{\label{l:mid3}
        $A_i \gets A_i \cup R$ \label{l:eq1_found_2}\;
        \Return{A}
    }\label{l:end}
    
    Let $g \in R$ be such that $v_p(A_p \cup g) \geq \mu$ \label{l:give_good}\\
    $A_p \gets A_p \cup g$ \label{l:while_loop_end}\;
    $R \gets R \setminus g$
    
}
\Return{$\A=(A_1, A_2, \ldots A_n)$} \label{l:finish}
\caption{Marginal-witness valuations.}
\label{alg:marginal_witness}
\end{algorithm}

%% file: sections/nonneg.tex
\section{Non-negative Valuations}\label{sec:nonneg}

In this section, we show the existence of \EQ{1} allocations for nonnegative valuations. Recall that a valuation function is \emph{nonnegative} if every subset of items has a nonnegative value. Such valuations need not be monotone: the marginal value of adding an item to a bundle can be negative. That is, for some agent \(i\), item \(e\), and bundle \(S\), it may hold that \(v_i(S \cup \{e\}) < v_i(S)\). Nevertheless, the total value of any bundle is always nonnegative, i.e., \(v_i(S) \ge 0\) for all \(S \subseteq M\). Our result settles the open question from \cite{bilò2025}.

\begin{restatable}{theorem}{nonneg}
\label{thm:nonneg}
Given a fair division instance $\GenInstance$ with nonnegative valuations, an \EQ{1} allocation always exists. Furthermore, if $|M|\ge |N|$, then there exists an \EQ{1} allocation where each agent receives a non-empty bundle. 
\end{restatable}

The proof of Theorem~\ref{thm:nonneg} is constructive and is based on Algorithm~\ref{alg:global_non_neg}, which we describe next.


\paragraph{Algorithm Overview.}
We start with an empty allocation and a pool \(R\) containing all the items. In each step of our algorithm, we would like to select a subset \(S \subseteq R\) and an agent \(i\) that minimize \(v_i(A_i \cup S)\) over all agents and subsets, then give \(S\) to \(i\) and remove \(S\) from \(R\). Throughout this process, two invariants are maintained: (i) no agent is ever made poorer upon receiving a set—if \(v_i(A_i \cup S) < v_i(A_i)\), then at an earlier step when \(A_i\) was formed, the union \(A_i \cup S\) would have been a strictly better minimizer; and (ii) after a set \(S\) is given to an agent \(i\), she becomes one of the rich agents—if some other agent \(j\) were richer than \(i\) after \(i\) received \(S\), it would contradict the formation of \(A_j\). Furthermore, due to the minimality of \(v_i(A_i \cup S)\), every other agent \(j \neq i\) can be made (weakly) richer than \(i\) by adding any non-empty subset, and in particular, any single item from the pool. Consequently, if at some point, by sheer luck, there are exactly \(n-1\) items remaining in the pool \(R\), the process can be stopped, and each of the agents, except the last one to receive a set, can be given one item from the pool to obtain an \(\EQ{1}\) allocation, with a lower \(\EQ{1}\) witness being the value of the last agent who received a set.

Now, to ensure that the pool indeed ends with \(n-1\) items, Algorithm~\ref{alg:global_non_neg} restricts attention to \emph{valid} bundles \(S\) whose removal leaves at least \(n-1\) items in the pool, i.e., \(|R \setminus S| \ge n-1\). Lemma~\ref{claim:rich_invariant} shows that the two invariants continue to hold under this validity constraint. Hence, when the loop terminates, \(|R| = n-1\), and the remaining items can be assigned one per agent (except the last agent to receive some subset) to obtain an \(\EQ{1}\) allocation.

\input{figs_tables_algo/global_non_neg.tex}

We first prove a loop invariant stating that, the agent who receives a valid bundle in any iteration is a rich agent at the end of that iteration. 
\begin{lemma}[Loop Invariant]\label{claim:rich_invariant}
    Let $\A^t$ be the allocation at the end of the $t^\text{th}$ iteration of the while loop in Line~\ref{step:while_nonneg} of Algorithm~\ref{alg:global_non_neg}. Suppose agent $i$ receives a bundle $S_t$ in this iteration. Then $i$ is a rich agent in $\A^t$, i.e.,
    \[
        v_i(A^t_i) \;\geq\; v_j(A^t_j) \quad \forall j \in N.
    \]
\end{lemma}

\begin{proof}
    We proceed by induction on $t$. Let $S_t$ and $R_t$ denote the bundle allocated in iteration $t$ and the pool at end of iteration $t$, respectively. Let $\A^t = (A^t_1,\dots,A^t_n)$ be the allocation at the end of iteration $t$.  

    \paragraph{Base case ($t=1$).} At the start, all agents have empty bundles, so $v_j(A^0_j)=0$ for all $j\in N$. In the first iteration, some agent $i$ receives a valid bundle $S_1$. Since valuations are nonnegative, $v_i(A^1_i)=v_i(S_1)\ge 0$, while every other agent still has a value of $0$. Hence $i$ is a rich agent in $\A^1$.  

    \paragraph{Inductive step.} Assume the claim holds at the end iteration $t-1$. That is, if agent $i$ received a bundle $S_{t-1}$ in iteration $t-1$, then $i$ is rich in $\A^{t-1}$. Let $\theta = v_i(A^{t-1}_i)$ be the value of the bundle of agent $i$ at this point.  

    Now consider iteration $t$, in which some agent $j$ receives $S_t$. We distinguish two complementary cases, depending on whether $j\ne i$ or $j=i$.
    \vskip0.2cm
    \textbf{Case A:} Agent $j\ne i$. It suffices to show that $v_j(A^t_j) \geq \theta$ as no other agent's bundle changes in this iteration. By the choice of $(j,S_t)$ in iteration $t$, we have
    \begin{align*}
        v_j(A^t_j) &= v_j(A^{t-1}_j \cup S_t) \nonumber \\
        &= v_j(A_j^{t-2} \cup S_t)\\
                   &\ge \min\limits_{\substack{\ell\in N\ \\ \text{valid }S \subset R^{t-2}}} v_\ell(A^{t-2}_\ell \cup S) \tag{since $S_t$ is also a valid set in iteration $t-1$} \\
                   &= v_i(A^{t-2}_i \cup S_{t-1}) \tag{since $(i,S_{t-1})$ is the minimizer in iteration $t-1$} \\
                   &= \theta 
    \end{align*}
    Hence, agent $j$ is rich in $\A^t$.

    \vskip0.2cm
    \textbf{Case B:} Agent $j=i$. That is, agent $i$ received a valid set $S_{t-1}$ in iteration $t-1$ and another valid set $S_t$ in iteration $t$. Note that since $S_t$ is valid at the start of iteration $t$, we have $|R_{t-1} \setminus S_t| \ge n-1$. Furthermore, since $R_{t-1} = R_{t-1}\setminus S_{t-1}$, we have $|R_{t-2} \setminus (S_{t-1} \cup S_t)| \ge n-1$, which implies that $S_{t-1} \cup S_t$ was a valid set at the start of iteration $t-1$. Therefore, we have

    \begin{align*}
        v_i(A^t_i) &= v_i(A^{t-1}_i \cup S_t) \nonumber \\
                   &= v_i(A^{t-2}_i \cup S_{t-1} \cup S_t) \\
                   &\ge \min\limits_{\substack{\ell\in N\ \\ \text{valid }S \subset R^{t-2}}} v_\ell(A^{t-2}_\ell \cup S) \tag{since $S_t \cup S_{t-1}$ was a valid set in iteration $t-1$} \\
                   &= v_i(A^{t-2}_i \cup S_{t-1}) \tag{since $(i,S_{t-1})$ is the minimizer in iteration $t-1$} \\
                   &= \theta
    \end{align*}

    Hence, agent $i$ is rich in $\A^t$. Therefore, by induction, the claim holds for all iterations $t$.
    
\end{proof}

Suppose agent $i$ receives a valid bundle $S$ in an iteration. By the above loop invariant, agent $i$ becomes a rich agent at the end of that iteration. Since $i,S$ is the minimizer over all $i\in N$ and valid subsets $S$ of $R$, any other agent who could have received a valid bundle in the same iteration would have obtained a value at least as large as that of agent $i$. We note this consequence as the following corollary.

\begin{corollary}
\label{cor:valid_bundle}
    Let $\mathbb{V}$ be the family of all valid subsets of the pool $R$ at the beginning of some iteration $t$. Let $\A^t$ be the allocation at the end of iteration $t$, where agent $j$ was allocated a valid subset $S$. Then, for any agent $i \in N\setminus\{j\}$ and any valid subset $T \in \mathbb{V}$,
    \[
    v_i(A^t_i \cup T) \ge v_j(A^t_j)
    \]
\end{corollary}

We now prove Theorem~\ref{thm:nonneg}.

\begin{proof}[Proof of Theorem~\ref{thm:nonneg}]
    Firstly, in each iteration of the while loop, we allocate a non-empty valid bundle to some agent. Since the pool starts with $m$ items and each valid bundle contains at most $m-(n-1)$ items, the while loop runs at most $m-(n-1)$ times. Thus, the algorithm terminates in a finite number of steps.

    Let $\A$ be the final allocation returned by the algorithm. If the instance had less than $n$ items, then the while loop in Line~\ref{step:while_nonneg} never runs. In this case, every agent gets at most one item, hence $\A$ is \EQ{1} with a lower \EQ{1} witness of $0$. So, consider the case when there are at least $n$ items.

    Suppose the while loop in Line~\ref{step:while_nonneg} runs for $t$ iterations and let 
    \(\A^t\) be the allocation at the end of the $t^\text{th}$ iteration. Let $\ell$ be the agent who received the last valid bundle $S_t$ in iteration $t$. By Lemma~\ref{claim:rich_invariant}, agent $\ell$ is a rich agent in $\A^t$. As agent \(\ell\) does not receive any more items the end of the while loop, we have \(A_\ell = A^t_\ell\) in the returned allocation $\A$. Let $\theta$ denote the final value of agent $\ell$, i.e., \(\theta = v_\ell(A_\ell)\). 

    Since \(t^\text{th}\) iteration was the last iteration of the while loop, and \(S_t\) was a valid bundle, the pool \(R\) contains exactly \(n-1\) items at the end of iteration \(t\). Now, to construct the final allocation \(A\), we give one item, say \(g_i\) from the pool to each agent \(i\) in \(N \setminus \{\ell\}\). Note that any singleton subset of \(R\) was a valid bundle at the beginning of iteration \(t\). So, by Corollary~\ref{cor:valid_bundle}, for any agent \(i \in N \setminus \{\ell\}\) and any item \(g \in R\), we have \(v_i(A^t_i \cup \{g\}) \ge v_\ell(A^t_\ell) = \theta\). Since agent \(i\) receives exactly one item \(g_i\) from the pool, we have \(v_i(A_i) = v_i(A^t_i \cup \{g_i\}) \ge \theta\). Also, $v_i(A_i \setminus \{g_i\}) = v_i(A_i^t) \leq \theta$. Hence, \(\theta\) is a lower \EQ{1} witness for the final allocation \(\A\). Thus, the algorithm always returns an \EQ{1} allocation.
\end{proof}

\input{sections/applications.tex}

%% file: figs_tables_algo/global_non_neg.tex
\begin{algorithm}[ht]
\SetAlgoLined
\DontPrintSemicolon
\KwIn{An instance $\GenInstance$ with non-negative valuations.} 
\KwOut{An \EQ{1} allocation.}
$A_i \gets \phi$ for all $i \in N$.\;
$R \gets M$\tcp*{Set of remaining items}
$\ell \gets $ an arbitrary agent\tcp*{ this variable will maintain the agent who received the last valid bundle}
\While{$|R| \geq n$}{\label{step:while_nonneg}
    $\mathbb{V}\gets \{S\subset R :  |R\setminus S|\ge n-1, S\ne \emptyset\}$\tcp*{Set of all valid bundles}\label{step:v_nonneg}
    $(\ell,S)\gets \argmin\limits_{j\in N, T\in \mathbb{V}} v_j(A_j\cup T)$\tcp*{Break ties arbitrarily}
    $A_\ell \gets A_\ell \cup S$\;
    $R \gets R\setminus S$\;
}
\For{$j\in N\setminus \{\ell\}$}{\label{step:richest_nonneg}
    Pick any item $g \in R$ arbitrarily\;
    $A_j \gets A_j \cup \{g\}$\;
    $R \gets R \setminus \{g\}$\;
}\label{l:non_neg_end}

\Return{$\allocn$}
\caption{Non-negative valuations.}
\label{alg:global_non_neg}
\end{algorithm}

%% file: sections/applications.tex
\subsection{Applications Beyond Fair Division}
\label{sec:graph_partitioning}
Many important functions in statistics and graph theory are nonnegative but not necessarily monotone. For example, statistical measures such as average, variance, and standard deviation, as well as graph-theoretic functions like the cut function and graph density, are all nonnegative and often non-monotone. Consequently, Theorem~\ref{thm:identical_subadd} has broad applicability across diverse domains. In this section, we highlight two notable applications, both of which have been previously explored in the context of fair division.

\subsection*{Equitable Graph Partitioning}
An interesting application of Theorem~\ref{thm:identical_subadd} is in equitable graph partitioning. Given an undirected graph \(G = (V, E)\) and an integer \(k \leq n\), the objective is to partition the vertex set \(V\) into \(k\) non-empty parts \(V_1, V_2, \ldots, V_k\) such that for any pair of parts \(V_i\) and \(V_j\), the difference in their cut values is at most \(\alpha\), i.e., \(|\delta(V_j) - \delta(V_i)| \leq \alpha\) for some \(\alpha \geq 0\). Here, \(\delta(V_i)\) denotes the cut value of \(V_i\), defined as the number of edges in \(E\) with exactly one endpoint in \(V_i\). This problem captures the problem of distributing the ``boundary'' edges as evenly as possible among the parts, and arises naturally in load balancing, network design, and parallel computing. The cut function is a classic example of a nonnegative, submodular, but non-monotone set function, making it a natural candidate for the application of our general existence results.

In~\cite{barman2025eq}, it was shown that for any undirected graph \(G\) and integer \(k\), there exists a partition of \(V\) into \(k\) non-empty parts such that the cut values differ by at most \(5\Delta + 1\), where \(\Delta\) is the maximum degree in \(G\). We improve this bound to \(\Delta\) by applying Theorem~\ref{thm:nonneg}. {{This improvement was also shown by \cite{hosseini25} by exhibiting the existence of EF1 and \emph{(weak) Transfer Stable} allocations for \emph{cut-valuations}}}. Likewise, we model the problem as a fair division instance with \(k\) agents, where the items are the vertices and each agent has the identical nonnegative cut function as their valuation. By applying Theorem~\ref{thm:nonneg}, we obtain an \EQ{1} partition with a lower witness threshold. This means that for any two parts \(V_i\) and \(V_j\) with \(\delta(V_i) < \delta(V_j)\), there exists a vertex \(u \in V_j\) such that \(\delta(V_j \setminus \{u\}) \leq \delta(V_i)\). Since adding or removing any vertex changes the cut value by at most \(\Delta\), it follows that \(\delta(V_j) \leq \delta(V_i) + \Delta\). Therefore, we achieve a partition into \(k\) non-empty parts such that the cut values differ by at most \(\Delta\), improving upon the previous bound.

While the existence result above depends on an exponential time algorithm( algorithm~\ref{alg:global_non_neg}), we can leverage additional structure in the cut function. Specifically, the cut function is both nonnegative and submodular. By modifying Algorithm~\ref{alg:marginal_witness}, we show that for any nonnegative submodular function, an \EQ{1} allocation with non-empty bundles can be found in polynomial time, if the number of items is greater than or equal to the number of agents.

\begin{restatable}{theorem}{graphs}
\label{thm:graphs}
Given a fair division instance \(\GenInstance\) with $
|M| \geq |N|$, where each valuation \(v_i\) is both nonnegative and submodular, there exists a polynomial-time algorithm that finds an \EQ{1} allocation in which every agent receives a non-empty bundle. Moreover, this allocation admits a lower \EQ{1} witness.
\end{restatable}

The modification to the algorithm and the proof of this theorem are provided in Appendix~\ref{app:submodular_polytime}. Applying this result to the graph partitioning problem, we obtain the following corollary.

\begin{corollary}
Given an undirected graph \(G = (V, E)\) with maximum degree \(\Delta\) and an integer \(k \leq |V|\), there exists a polynomial-time algorithm to partition \(V\) into \(k\) non-empty parts such that the cut values of the parts differ by at most \(\Delta\).
\end{corollary}

\subsection*{Uniformly Dense Graphs}

Given an undirected graph \(G = (V, E)\), the density \(\rho(S)\) of a subset \(S \subseteq V\) is defined as the ratio of the number of edges with both endpoints in \(S\), to \(|S|\), i.e., \(\rho(S) = \frac{|E(S)|}{|S|}\), where \(E(S)\) denotes the set of such edges. The density function is nonnegative, and the marginal contribution of any vertex to the density is at most 1~\cite{barman2025eq}.

\citeauthor{barman2025eq}~\cite{barman2025eq} showed that for any undirected graph and integer \(k\), the vertex set can be partitioned into \(k\) non-empty parts such that the densities of any two parts differ by at most $4$. As a direct consequence of Theorem~\ref{thm:nonneg}, we improve this factor to $1$.

\begin{corollary}
Given an undirected graph \(G = (V, E)\) and an integer \(k \leq |V|\), there exists a partition of \(V\) into \(k\) non-empty parts such that the densities of any two parts differ by at most $1$.
\end{corollary}

%% file: sections/identical_subadd.tex
\section{Identical Subadditive Valuations}
\label{sec:identical_subadditive}

We shall now consider a special case of subadditive valuations, where all agents have identical valuation functions. Given such an instance, we show that Algorithm~\ref{alg:identical_subadd} always returns an \EQ{1} allocation, provided that the common valuation function \(v(\cdot)\) is nonnegative for the grand bundle \(M\), i.e., \(v(M) \ge 0\).

\begin{restatable}{theorem}{identicalsubadd}
\label{thm:identical_subadd}
Given a fair division instance $\mathcal{I} = \langle N, M, \{v\}_{i \in N} \rangle$, where the agents share an identical subadditive valuation function and the grand bundle is valued nonnegatively, \Cref{alg:identical_subadd} returns an \EQ{1} allocation.
\end{restatable}

Note that when agents have identical valuations, an allocation that is \EQ{1} is also \EF{1}. Therefore, we have the following corollary.

\begin{restatable}{corollary}{identical}
\label{cor:identical}
Given a fair division instance $\mathcal{I} = \langle N, M, \{v\}_{i \in N} \rangle$, where the agents share an identical subadditive valuation function and the grand bundle is valued nonnegatively, \Cref{alg:identical_subadd} returns an \EF{1} allocation.
\end{restatable}

\paragraph{Algorithm Overview.} The core idea of our algorithm is a reduction to the nonnegative case as follows: we start by giving a largest cardinality subset \(L\) of items valued negatively, i.e., $v(L)<0$,  to an arbitrary agent $\ell \in N$. Now, from subadditivity of $v$, and the maximality of $L$, we get that any subset of the remaining items \(R = M \setminus L\) has a nonnegative value. If there are at most $n$ items left in \(R\), we simply give one more item to agent $\ell$ and distribute the remaining items, one each, to some of the other agents. On the other hand, if there are more than $n$ items in the pool, note that the current allocation (with agent $\ell$ having $L$ and the rest of the agents having the empty bundle) satisfies the loop invariant (informally that giving any agent a valid subset of the remaining items, makes it a rich agent) defined in Lemma~\ref{claim:rich_invariant}, and the set of valid subsets is non empty. Hence, we proceed exactly as in Algorithm~\ref{alg:global_non_neg}, that is, by repeatedly allocating a valid bundle \(S \subseteq R\) to an agent \(i\) that minimizes \(v_i(A_i \cup S)\) across all agents and valid bundles. We repeat this until there are no valid bundles left, and then give one remaining item to each of the other agents, except the one who received the last valid bundle.

\input{figs_tables_algo/algo_identical_subadd.tex}

\begin{proof}[Proof of \Cref{thm:identical_subadd}]
    Consider first the case when $|R| > n$. In this case, we simply run the while loop from algorithm~\ref{alg:global_non_neg}. Note from the proof of Theorem~\ref{thm:nonneg} that, running the while loop of ~\ref{alg:global_non_neg} starting with a partial allocation $A$ yields an \EQ{1} allocation if the following two conditions are satisfied (where $R = M \setminus \cup_{i \in N} A_i$):

    \begin{enumerate}
        \item $|R| \geq n$.
        \item The invariant from Lemma~\ref{claim:rich_invariant} holds, that is, $v_i(A_i \cup S) \geq \max_{i \in N} v_i A_i$, for all $i \in N$, non-empty subsets $S \subseteq R$, such that $|R \setminus S| \geq n - 1$.
    \end{enumerate}

    Our invocation of the while loop of Algorithm~\ref{alg:global_non_neg} begins with the allocation $A_\ell = L$, $A_i = \emptyset$ for all $i \in N \setminus \ell$, and $|R| > n$. The second condition is satisfied because the right hand side ($\max_{i \in N} v_i(A_i)$) is $0$. Now, consider a non-empty subset $S \subseteq R$. By subadditivity, $v(L) + v(S) \geq v(L \cup S)$.  Also, $v(L \cup S) \geq 0$ (and hence the second condition is satisfied for agent $\ell$), since $L$ a maximum cardinality subset with a negative value. Thus, $v(S) \geq v(S \cup L) - v(L) > 0$ as $v(L) < 0$, and hence the second condition also holds for agents $i \in N \setminus \ell$.

    In the case when $|R| \le n$, we give one item to $\ell$ (note that $|R| > 0$ since $L \neq M$, as $v(M) \ge 0$), and then we give one item each to some of the other agents. It is trivial to see that $0$ is a lower \EQ{1} witness for the final allocation. This is because $v(L) < 0$, and $v(L \cup \{g\}) \geq 0$ for all $g \in R$, by the maximality of $L$. Also, $v(\{g\}) > 0$ for all $g \in R$, since $v(S) > 0$ for all non-empty $S \subseteq R$.

   \end{proof}

%% file: figs_tables_algo/algo_identical_subadd.tex
\begin{algorithm}[ht]
\SetAlgoLined
\DontPrintSemicolon

\KwIn{An instance $\mathcal{I} = \langle N, M, v\rangle$ with identical subadditive valuation $v$, where $v(M) \geq 0$.}

\KwOut{An \EQ{1} allocation.}
$A_i \gets \phi$ for all $i \in N$.\;
$R \gets M$\tcp*{Set of remaining items}
$L \gets \argmax\limits_{S\subseteq M, v(S)<0} |S|$\tcp*{A largest cardinality subset with negative value}\label{step:largest_neg}
$\ell \gets $ an arbitrary agent\\
$A_\ell \gets L$\tcp*{Give $L$ to an arbitrary agent $\ell$}
$R \gets R \setminus L$\tcp*{Remove $L$ from the pool of items}
\If{$|R| \le n$}{
    $g \gets $ an arbitrary item from $R$\\
    $A_\ell \gets A_\ell \cup \{g\}$\\
    $R \gets R \setminus \{g\}$\\
    $X \gets N\setminus \{\ell\}$\\
    \While{$R \neq \emptyset$}{
        $j \gets$ an arbitrary agent from $X$\;
        $g \gets$ an arbitrary item from $R$\;
        $A_j \gets A_j \cup \{g\}$\;
        $R \gets R \setminus \{g\}$\;
        $X \gets X \setminus \{j\}$
    }
}

\Else{

    Run lines~\ref{step:while_nonneg} to ~\ref{l:non_neg_end} of Algorithm~\ref{alg:global_non_neg}
}

\Return{$\allocn$}
\caption{Identical subadditive valuations.}
\label{alg:identical_subadd}
\end{algorithm}

%% file: sections/eq_transform.tex
\section{The Case of nonpositive Valuations}
\label{sec:nonpos}

Until now, we considered fair division instances where the agents have non-monotone valuations but the grand bundle is nonnegatively valued by all agents. In this section, we show that all our positive results for \EQ{1} allocations extend to the case when the grand bundle is nonpositively valued by all agents, i.e., \(v_i(M) \leq 0\) for all \(i \in N\). Towards this, we prove a more general result that shows that \EQ{1} is preserved when the instance is transformed by negating the valuation functions of all agents. Formally stated as follows:

\begin{lemma}\label{thm:nonpos_to_nonneg}
    Given a fair division instance $\GenInstance$, an allocation $\A$ is \EQ{1} for $\GenInstance$ if and only if it is \EQ{1} for the instance $\calI' = \langle N, M, \{-v_i\}_{i \in N} \rangle$.
\end{lemma}

\begin{proof}
    Suppose $\A$ is an \EQ{1} allocation for $\calI$. We shall show that $\A$ is also \EQ{1} for $\calI'$. Consider any two agents $i, j \in N$. If $-v_i(A_i) \ge -v_j(A_j)$, then $i$ is already equitable with respect to $j$ in $\calI'$.

    Now, suppose $-v_i(A_i) < -v_j(A_j)$. Then we have $v_j(A_j) < v_i(A_i)$. That is, $j$ is not equitable towards $i$ in $\calI$. Since $\A$ is \EQ{1} for $\calI$, there exists $S\subseteq A_j\cup A_i$, $|S|\le 1$ such that $v_j(A_j\setminus S) \ge v_i(A_i\setminus S)$. This implies that $-v_i(A_i\setminus S) \ge -v_j(A_j\setminus S)$. Thus, $i$ is equitable towards $j$ in $\calI'$ after removing $S$. Hence, $\A$ is \EQ{1} for $\calI'$.

    The argument is symmetric in the other direction. Thus, $\A$ is \EQ{1} for $\calI$ if and only if it is \EQ{1} for $\calI'$.
\end{proof}

Therefore, to show the existence of \EQ{1} allocations for the case of nonpositively valued grand bundle, it suffices to show the existence of \EQ{1} allocations for the case of nonnegatively valued grand bundle. Similarly, to show hardness results for nonpositively valued grand bundle, it suffices to show hardness results for nonnegatively valued grand bundle.

It is relevant to note here that such translations do not work in general for other fairness notions like proportionality, envy-freeness, and their approximations (Prop/PropX/EF/\EF{1}). An allocation is said to be Proportional (Prop) for goods (chores) if every agent ends up receiving at least (at most) $\frac{1}{n}$ of its value for the entire set of items. An allocation is PropX for goods (chores) if the addition (removal) of any one item suffices to achieve Proportionality. A PropX allocation may not even exist for goods \cite{AZIZ2020} but always exists for chores and is known to be efficiently computable for additive valuations \cite{Li2022}.
While the classical envy-cycle elimination algorithm \cite{LMM04} produces an \EF{1} allocation for the case of positively-valued items with monotone valuations, the same fails to work for chores \cite{Bhaskar2020}. 

An intuitive justification for why this holds for \EQ{1}, could be that the notion of equitability is inherently symmetric, unlike envy-freeness and proportionality. That is, an envy-free allocation (similarly a proportional allocation) can in some sense be \emph{more} than envy-free, where some agent value their bundle strictly more than another agent's bundle. When the valuations are negated, this relationship is reversed, and the allocation may no longer be envy-free. However, in an equitable allocation, all agents value their bundles equally, and negating the valuations preserves this equality.

Due to \Cref{thm:nonpos_to_nonneg}, all our positive and negative results for \EQ{1} allocations in previous sections extend to the case when the grand bundle is nonpositively valued by all agents. We summarize these results below:

\begin{corollary}
    \label{thm:two_agents_nonpos}
    Given a fair division instance $\GenInstance$ with two agents such that the grand bundle is valued nonpositively, an \EQ{1} allocation always exists and can be computed efficiently.
\end{corollary}

\begin{proof}
    Given a fair division instance $\GenInstance$ with two agents such that the grand bundle is valued nonpositively, consider the instance $\calI' = \langle N, M, \{-v_i\}_{i \in N} \rangle$. Note that in $\calI'$, the grand bundle is valued nonnegatively by both agents. By Theorem~\ref{thm:two_agents}, there exists an \EQ{1} allocation $\A$ for $\calI'$, which can be computed efficiently. By \Cref{thm:nonpos_to_nonneg}, $\A$ is also an \EQ{1} allocation for $\GenInstance$.
\end{proof}

Similarly, due to \cref{thm:nonpos_to_nonneg} and Theorem~\ref{thm:nonneg}, we have the following corollary.

\begin{corollary}\label{thm:nonpos}
Given a fair division instance $\GenInstance$ with nonpositive valuations, an \EQ{1} allocation always exists. Furthermore, if $|M|\ge |N|$, then each agent can be guaranteed a non-empty bundle. 
\end{corollary}

Furthermore, we know that if $v$ is a submodular function, then $-v$ is supermodular. Therefore, we can state the following corollaries. The proofs of the below corollaries follow from \cref{thm:nonpos_to_nonneg} combined with Theorem~\ref{thm:submodular} and Theorem~\ref{thm:doublymonotone} respectively.

\begin{corollary}\label{thm:supermodular_nonpos}
Given a fair division instance $\GenInstance$ with supermodular valuations and every agent values the grand bundle nonpositively, an \EQ{1} allocation always exists and can be computed in polynomial time.
\end{corollary}

\begin{corollary}\label{thm:doublymonotone_nonpos}
Given a fair division instance $\GenInstance$ with doubly monotone valuations and every agent values the grand bundle nonpositively, an \EQ{1} allocation always exists and can be computed in polynomial time.
\end{corollary}

Similarly, if a valuation function is subadditive, then its negation is superadditive. Therefore, we have the following corollary, due to \cref{thm:nonpos_to_nonneg} and Theorem~\ref{thm:identical_subadd}.

\begin{corollary}\label{thm:superadditive_nonpos}
Given a fair division instance $\GenInstance$ with identical superadditive valuations such that the grand bundle is valued nonpositively, an \EQ{1} allocation always exists. 
\end{corollary}

Similarly, the hardness result from Theorem~\ref{thm:hard_submodular} extends to the case when the grand bundle is nonpositively valued by all agents.

\begin{corollary} \label{thm:hard_submodular} For valuations where each agent values the grand bundle nonpositively, it is NP-Hard to decide whether an \EQ{1} allocation exists, even when the valuations are submodular.
\end{corollary}

%% file: sections/conclusion.tex
\section{Conclusion}\label{sec:conclusion}

In this work, we studied the problem of finding an equitable up to one item (\EQ{1}) allocation for indivisible items under general valuations. Our findings resolve open questions in equitable division and extend the existential guarantees to much richer valuation models, under the reasonable assumption that agents agree on the sign of their value for the grand bundle. While our algorithm for submodular valuations and doubly monotone valuations runs in polynomial time, our algorithm for nonnegative valuations takes exponential time. Whether this latter case admits an efficient algorithm remains an open question. Another interesting question is whether \EQ{1} exists for non-identical subadditive valuations. Exploring the stronger relaxation of equitability up to \emph{any} item (EQX) in this context is also an immediate direction. Finally, characterizing instances where such allocations are compatible with efficiency guarantees presents an exciting direction for future research.

%% file: sections/acknowledgements.tex
\section*{Acknowledgements}
HH acknowledges the support from the National Science Foundation (NSF) through CAREER Award IIS-2144413 and Award IIS-2107173. HV acknowledges the support from TCS RSP Fellowship. AS acknowledges the support from Walmart Center for Tech Excellence (CSR WMGT-23-0001). JY acknowledges the support from Google PhD fellowship.

%% file: sections/appendix.tex
\clearpage
\appendix
{\centering\LARGE{\textbf{Appendix}}}

\section{Missing Proofs from Section \ref{sec:general_agents}}
\label{sec:hard_appendix}
Our proof of~\Cref{thm:hard_supermodular} uses a reduction from \RestrictedPartition, which we define below. 

\begin{definition}[\RestrictedPartition]
        In this problem, we are given an input multiset $M' = \{b_1,\ldots,b_m\}$ of positive integers with total sum $T = \sum_{i=1}^m b_i$ and each element strictly less than one quarter of the total, i.e., $0 < b_i < T/4$ for all $i \in [m]$. The goal is to decide whether there exists a partition of the index set $[m]$ into two disjoint subsets $S_1, S_2$ (so $S_1 \cup S_2 = [m]$, $S_1 \cap S_2 = \emptyset$) such that
        \[
        \sum_{i \in S_1} b_i = \sum_{j \in S_2} b_j = T/2.
        \]
        A yes-instance admits such a partition; otherwise it is a no-instance. Note that \RestrictedPartition{} differs from the standard \Partition{} problem only in the constraint $b_i < T/4$ on the input.
    \end{definition}

\begin{lemma} \label{lem:rest-part}
\RestrictedPartition{} is NP-Hard.
\end{lemma}

\begin{proof}
    We give a reduction from the \Partition{} problem. Suppose that we are given $n$ integers $a_1, a_2, \ldots a_n$ and we need to find if there exists a subset with sum equal to half of the total sum $T = \sum_{i \in n} a_i$. Let us add $4$ copies of $T$ to $a$ to create a list $b$ with $n + 4$ integers, such that $b_i = a_i$ for all $i \in [n]$ and $b_i = T$ for all $i \in \{n+1,n+2,n+3,n+4\}$. Clearly, the total sum of $b$ equals $5T$ and $b$ is a valid instance of the \RestrictedPartition, since every value in $b$ is at most $T$, which is strictly smaller than a quarter of the total sum. 
    
    If the original instance was a ``Yes'' instance of the partition problem, then $ \exists S \subseteq [n]$ such that $\sum_{i \in S} a_i = \frac{T}{2}$. Hence, $\sum_{i \in S \cup \{n+1,n+2\}} b_i = \frac{5T}{2}$ and the reduced instance is a ``Yes'' instance of the restricted partition problem.
    
    If the reduced instance is a ``Yes'' instance of the restricted partition problem, then $\exists S \subseteq [n+4]$ such that $\sum_{i \in S} b_i = \frac{5T}{4}$. Let $c = |S \cap \{n+1,n+2,n+3,n+4\}|$. Then, $(\sum_{i \in S \cap [n]} a_i) + c \cdot T = 2T + \frac{T}{2}$. Now, $c$ must be equal to $2$ since $\sum_{i \in S \cap [n]} a_i \in [0, T]$. Hence, $\sum_{i \in S \cap [n]} a_i = \frac{T}{2}$, and the original instance must have been a ``Yes'' instance of the partition problem.
\end{proof}

We are now ready to prove \Cref{thm:hard_supermodular}.

\hardness*

\begin{proof} Given a \RestrictedPartition~instance in the form of a multiset $M' = \{b_1, \ldots, b_m\}$ where $m \geq 5$, we construct a fair division instance $\calI$ as follows. We create $n=3$ agents and a set $M = \{1, 2, \ldots, m\}$ of $m$ items. Agents $1$ and $2$ have identical valuations. In particular, for any subset $S \subseteq M$, we have

    \[
    v_1(S) = v_2(S) = \begin{cases} 0 & \text{if } S = \emptyset \\
    2\left(\sum_{j \in S} b_j\right) - T, & \text{ otherwise} \\
    
\end{cases}
\]

For agent $3$, we have $v_3(S) = |S|$ for all $S \subseteq M$.

Note that $\{v_i\}_{i \in [n]}$ value the grand bundle ($M$) nonnegatively, and are supermodular. The non-negativity is immediate since $v_1(M) = v_2(M) = 2T - T = T > 0$ and $v_3(M) = m > 0$. The submodularity of $v_3$ follows from the fact that $v_3$ is additive. To see that $v_1$ and $v_2$ are supermodular, we consider any agent $i \in \{1, 2\}$, consider any subsets $S \subset S' \subseteq M$ and an item $e \in M \setminus S'$. Clearly, $v_i(S' \cup e) - v_i(S') = 2b_e$ since $S' \neq \emptyset$. Also, $v_i(S \cup e) - v_i(S) = 2 b_e - T$ if $S = \emptyset$ and $2b_e$ otherwise. In both cases, we have $v_i(S' \cup e) - v_i(S') \geq v_i(S \cup e) - v_i(S)$.

This completes the construction of the fair-division instance $\calI$. We now show that $M'$ is a yes-instance of \RestrictedPartition~if and only if $\calI$ is a yes-instance of deciding the existence of an \EQ{1} allocation.

\paragraph{Forward Direction.} Suppose $M'$ is a yes-instance of \RestrictedPartition. Then, there exists a partition $[m] = S_1 \cup S_2$ such that $\sum_{i \in S_1} b_i = \sum_{j \in S_2} b_j = \frac{T}{2}$. Consider the following allocation $\calA$ where $A_1 = S_1, A_2 = S_2$ and $A_3 = \emptyset$. Note that $v_i(A_i) = 0$ for all $i \in [3]$ and, therefore, $\calA$ is an \EQ{1} allocation.

\paragraph{Reverse Direction.} Suppose the fair-division $\calI$ instance is a yes-instance and let $\calA = (A_1, A_2, A_3)$ be the corresponding \EQ{1} allocation. We will show that $M'$ is a yes-instance of \RestrictedPartition. We will say that an ordered pair of agents $(i, j)$ violates \EQ{1}, $v_i(A_i) < v_j(A_j)$ and for all $e \in A_i \cup A_j$, $v_i(A_i \setminus e) < v_j(A_j \setminus e)$. Clearly, since $\calA$ is an \EQ{1} allocation, no ordered pair of agents violates \EQ{1}.

We consider the following cases:

\begin{enumerate}

\item $A_1 = \emptyset$ or $A_2 = \emptyset$. Suppose $A_1 = \emptyset$. If $|A_3| > 1$, then $v_3(A_3 \setminus e) > 0$ for all $e \in A_3$ and the pair $(1,3)$ violates \EQ{1}. Hence, $|A_3| \leq 1$. This implies that $|A_2| \geq m-1$ since $A_1 = \emptyset$. For any $e \in A_2$ (note that $A_2 \setminus e \neq \emptyset$ since $m \geq 5$), we have
 \begin{align*} 
    v_2(A_2 \setminus e) &= 2 \left( \sum_{k \in A_2 \setminus e} b_k \right) - T \\
    & = 2 \left( \sum_{k \in M} b_k - b_e - \sum_{f \in A_3} b_f \right) - T && \because A_1 = \emptyset\\
    & > 2(T - T/4 - T/4) - T && \because |A_3| \leq 1, \max_{e \in M} b_e < T/4\\
    & = 0
 \end{align*}

 Hence, the pair $(1,2)$ violates \EQ{1}, a contradiction. Therefore, $A_1 \neq \emptyset$. By symmetry, we have $A_2 \neq \emptyset$. Therefore, we have

 \begin{align*}
    v_1(A_1) + v_2(A_2) & = 2 \left( \sum_{e \in A_1} b_e \right) - T + 2 \left( \sum_{f \in A_2} b_f \right) - T  \\
    & =  2 \left( \sum_{k \in M} b_i - \sum_{g \in A_3} b_g \right) - 2T \\
    & = -2 \left( \sum_{g \in A_3} b_g \right)
\end{align*}



 In the rest of the analysis, we will assume without loss of generality that $v_1(A_1) \leq v_2(A_2)$.

\item Suppose $A_1, A_2, A_3 \neq \emptyset$. Then, from the above equation, we have $2 v_1(A_1) \leq - \left( \sum_{k \in A_3} b_k \right) < 0$. Note that $v_3(A_3) > 0$ and $v_3(A_3 \setminus e) \geq 0 > v_1(A_1)$ for all $e \in A_3$. Similarly, $v_1(A_1 \setminus e) \leq 0 < v_3(A_3)$ for all $e \in A_1$. Thus, pair $(1,3)$ violates \EQ{1}, a contradiction.

\item $A_1, A_2 \neq \emptyset$ and $A_3 = \emptyset$. In this case, we have $v_1(A_1) + v_2(A_2) = 0$. If $v_1(A_1) = v_2(A_2) = 0$.  First, suppose $v_1(A_1) < 0, v_2(A_2) = -v_1(A_1) > 0$. Then, we consider the following two sub-cases:

\begin{enumerate}
    \item $|A_1| > 1$. In this case, the pair $(1,3)$ violates \EQ{1} since $v_1(A_1) < 0 = v_3(A_3)$ and $v_1(A_1 \setminus e) = v_1(A_1) - 2b_e < 0 = v_3(A_3)$ for all $e \in A_1$.
    \item $|A_1| = 1$. This implies $|A_2| = m-1$. Then, for any $e \in A_2$ (note that $A_2 \setminus e \neq \emptyset$ since $m \geq 5$), we have
 \begin{align*} 
    v_2(A_2 \setminus e) &= 2 \left( \sum_{j \in A_2 \setminus e} b_j \right) - T \\
    & = 2 \left( \sum_{j \in M} b_j - b_e - \sum_{f \in A_1} b_f \right) - T &&  (A_3 = \emptyset)\\
    & > 2(T -   T/4 - T/4) - T &&  (|A_1| = 1, \max_{e \in M}b_e < T/4)\\
    & = 0
 \end{align*}
    Hence, the pair $(1,2)$ violates \EQ{1}, a contradiction.
\end{enumerate}

So, it must be that $v_1(A_1) = v_2(A_2) = 0$. Therefore, \RestrictedPartition~ is a yes-instance with $S_1= A_1, S_2 = A_2$. Indeed, $$v_1(A_1) = 2 \left( \sum_{i \in A_1} b_i \right) - T = 0 \Rightarrow \sum_{i \in A_1} b_i = \frac{T}{2} $$

\paragraph{Explicit Non-Existence Instance.}
Consider the \RestrictedPartition{} instance $M'=\{1,1,1,1,1\}$. Here $T=5$ and each $b_i=1 < T/4=1.25$, so it is a valid restricted instance. Since $T$ is odd, there is no partition of $M'$ into two subsets of equal sum; hence it is a no-instance.

Applying the reduction, we obtain a fair-division instance with $m=5$ items and three agents. For agents $1,2$ the valuation becomes
\[
v_1(S)=v_2(S)=
\begin{cases}
0 & S=\emptyset,\\
2|S|-5 & S\neq \emptyset,
\end{cases}
\qquad
v_3(S)=|S|.
\]
By the proven correctness of the reduction (\EQ{1} exists iff the restricted partition instance is a yes-instance), this fair-division instance admits no \EQ{1} allocation. This gives a concrete $3$-agent supermodular example (all value the grand bundle nonnegatively) witnessing non-existence of \EQ{1}. Also, note that, by introducing $n-3$ additional copies of agent $3$, we get that, for any $n \geq 3$, there exists an instance with $n$ agents having supermodular valuations with a nonnegative value for the grand bundle, where an \EQ{1} allocation does not exist.

\end{enumerate}
\end{proof}

\section{Missing proofs from Section~\ref{sec:submodular_dmonotone}} \label{apndx:submod_dmon}

In this section, we show that both submodular and doubly monotone valuations satisfy marginal witness property.

\submodular*
\begin{proof}
    It suffices to prove that submodular valuations satisfy the marginal witness property. Consider a submodular valuation function $v: 2^M \to \mathbb{R}$. Let $A, B \subseteq M$ be two disjoint subsets such that $B = \{b_1, b_2, \ldots b_k\} \neq \phi$ and $v(A\cup B) \geq v(A)$.

    Then, 

    \begin{align*}
    0 & \leq v_i(A \cup B) - v_i(A)\\
    &= \sum_{j=1}^{k} \left( v_i(A \cup \{b_1, \ldots b_j\}) - v_i(A \cup \{b_1, \ldots b_{j-1}\}) \right)\\
    & \leq \sum_{j=1}^{k} \left( v_i(A \cup b_j) - v_i(A)\right)\\
    & \leq t \cdot \max_{j=1}^{t} \left( v_i(A \cup b_j) - v_i(A)\right)
\end{align*}
\end{proof}

\doublymonotone*
\begin{proof}
    It suffices to prove that doubly monotone valuations satisfy the marginal witness property. Consider a doubly monotone function $v: 2^M\to\mathbb{R}$, and two disjoint subsets $A, B \subseteq M$, such that $B = \{b_1,\ldots b_k\} \neq \phi$ and $v(A\cup B) \geq v(A)$.

    If $v(A \cup b_i) < v(A)$ for all $i \in \{1, 2, \ldots k\}$, then, since $v$ is doubly monotone, $v(X \cup b_i) \leq v(X) $ for all $X \subseteq M$. Then, we arrive at a contradiction:

    \begin{align*}
        v(A \cup B) &= v(A \cup \{b_1\}) + \sum_{i=2}^{k} (v((A \cup \{b_1, \ldots b_{i-1}\}) \cup \{b_i\}) - v(A \cup \{b_1, \ldots b_{i-1}\} )\\
        & < v(A) + \sum_{i=2}^{k} 0 = v(A).
    \end{align*}
\end{proof}
\section{Missing Proofs from Section~\ref{sec:nonneg}}\label{app:submodular_polytime}

\graphs*

\begin{proof}
     We first modify the Algorithm~\ref{alg:marginal_witness} as follows: before starting the while loop, allocate one item each to the agents (the choice of the $n$ items allocated does not matter). With this modification, note that the invariants \textbf{Invariant 1} and \textbf{Invariant 2} defined in ~\Cref{lemma:inv_marg} hold true for this initial allocation. Hence, the result follows. 
\end{proof}